\documentclass[conference,a4paper]{IEEEtran}

\IEEEoverridecommandlockouts

\usepackage{balance}
\usepackage{amsmath,amssymb,amsthm}
\usepackage{url}
\usepackage{listings}
\usepackage{xcolor}
\usepackage[acronym]{glossaries}

\newacronym{milac}{MiLAC}{microwave linear analog computer}
\newacronym{dft}{DFT}{discrete Fourier transform}
\newacronym{fft}{FFT}{fast Fourier transform}
\newacronym{ai}{AI}{artificial intelligence}

\newtheorem{definition}{Definition}
\newtheorem{theorem}{Theorem}

\definecolor{leanblue}{RGB}{0,80,160}

\lstdefinelanguage{Lean}{
  morekeywords={abbrev, def, noncomputable, fun, _ , theorem, by, inductive, Prop, where},
  sensitive=true,
  morecomment=[l]{--},
  morecomment=[n]{/-}{-/},
  morestring=[b]"}

\begin{document}
\bstctlcite{BSTcontrol}

\title{Toward a Lean Formalization of\\Analog Computing with Microwaves}

\author{\IEEEauthorblockN{
Matteo Nerini\IEEEauthorrefmark{1},
Xuekang Liu\IEEEauthorrefmark{2},
Bruno Clerckx\IEEEauthorrefmark{1}
}
\IEEEauthorblockA{\IEEEauthorrefmark{1}
Department of Electrical and Electronic Engineering, Imperial College London, London, UK, m.nerini20@imperial.ac.uk}
\IEEEauthorblockA{\IEEEauthorrefmark{2}
School of Engineering, Lancaster University, Lancaster, UK}
\thanks{This work has been supported in part by UKRI under Grant EP/Y004086/1, EP/X040569/1, EP/Y037197/1, EP/X04047X/1, EP/Y037243/1.}}

\maketitle

\begin{abstract}
Analog computing with microwave signals can perform linear transformations directly in the analog domain, as the signals propagate through a microwave network.
A fundamental question is which transformations can be computed with a given set of microwave components.
In our previous work, we answered this question for networks of hybrid couplers and phase shifters by deriving a necessary and sufficient condition on the transformations these networks can compute, and we showed that the \gls{dft} satisfies it.
In this paper, we take a first step toward the formalization of analog computing with microwaves in Lean, a programming language and proof assistant that is increasingly adopted in mathematics.
We formalize the considered components, their series and parallel connections, and the class of networks they can implement.
Then, we formally prove the necessary and sufficient condition characterizing these networks, as well as the implementability of the \gls{dft} of any size power of two.
All proofs are checked by the Lean kernel, and the code is openly available at:
\\\url{https://github.com/matteonerini/formalizing-analog-computing}.
\end{abstract}

\vskip0.5\baselineskip

\section{Introduction}

A fundamental problem in analog computing is to characterize all the functions that can be computed with a specific physical system.
In a seminal work, Claude Shannon addressed this problem for the differential analyzer, the mechanical analog computer built by Vannevar Bush at MIT, and showed which functions it could generate \cite{sha41}.
Characterizations of this kind make it possible to understand what a class of analog computers can and cannot do.

The same problem has recently gained relevance for analog computing with microwaves, which is emerging as a complement to digital computing for latency- and energy-critical applications \cite{cal13,sil14,zan21,dar26,del26}.
In general, linear microwave networks can compute a matrix-vector product as the input signals propagate through them, without any digital operation \cite{ner25-1}.
These networks, denoted as \glspl{milac}, can compute specific linear transformations that depend on the components they are built from.
In \cite{ner26}, we characterized all the linear transformations computable by \glspl{milac} made of hybrid couplers and phase shifters, through a necessary and sufficient condition, and we showed that they include the \gls{dft}, Hadamard, and Haar transforms.

Results of this kind are valid over all the networks, of any size and topology, that can be made from the given components.
Therefore, they cannot be confirmed by simulations or measurements, which can only test individual networks.
This makes analog computing a natural target for \textit{formalization}, i.e., for writing definitions and proofs in a proof assistant, whose kernel checks every logical step.
Formalization is gaining momentum in mathematics, where proof assistants such as Lean \cite{dem21} are being used to verify major results \cite{buz25}, increasingly with the help of \gls{ai} (see the Navier-Stokes Millennium Prize Problem \cite{ope26}).
Beyond verification, formalization can accelerate research by turning known results into a machine-readable body of knowledge, on which new results can be reliably built.

In this paper, we take a first step toward the formalization of analog computing with microwaves by formalizing in Lean the results of \cite{ner26}.
Our contributions are as follows.
\textit{First}, we formalize the definitions of hybrid couplers, phase shifters, their series and parallel connections, and the class of networks they can implement (Section~III).
\textit{Second}, we formally prove the necessary and sufficient condition characterizing these networks (Section~IV).
\textit{Third}, we formally prove that the \gls{dft} of any size power of two is implementable with hybrid couplers and phase shifters (Section~V).
The complete Lean code is openly available, and all proofs are checked by the Lean kernel.
To the best of our knowledge, this is the first formalization of analog computing with microwave networks.

\section{Analog Computing with\\Hybrid Couplers and Phase Shifters}

In this section, we recall the \gls{milac} model and characterize the networks that can be built with hybrid couplers and phase shifters, as derived in \cite{ner26}.

A \gls{milac} is a $2N$-port linear microwave network, where the first $N$ ports are the input ports and the last $N$ ports are the output ports.
The $N$ input ports are driven by voltage sources $\mathbf{u}\in\mathbb{C}^{N\times 1}$ with series impedance $Z_0$, e.g., $Z_0=50~\Omega$, while the output voltages $\mathbf{v}\in\mathbb{C}^{N\times 1}$ are read at the $N$ output ports by probes with input impedance $Z_0$.
The $2N$-port \gls{milac} network can be characterized by its scattering matrix $\mathbf{S}\in\mathbb{C}^{2N\times 2N}$ \cite[Chapter~4]{poz12}, partitioned as
\begin{equation}
\mathbf{S}=
\begin{bmatrix}
\mathbf{S}_{11} & \mathbf{S}_{12}\\
\mathbf{S}_{21} & \mathbf{S}_{22}
\end{bmatrix},\label{eq:S}
\end{equation}
where $\mathbf{S}_{XY}\in\mathbb{C}^{N\times N}$, for $X,Y\in\{1,2\}$.
Following circuit theory \cite[Section~II]{ner26}, the output of the \gls{milac} is $\mathbf{v}=\mathbf{S}_{21}\mathbf{u}/2$, i.e., it is a linear function of the input depending on the block $\mathbf{S}_{21}$.
Thus, characterizing the functions computable by a \gls{milac} built from given components corresponds to characterizing the blocks $\mathbf{S}_{21}$ that these components can realize.

\subsection{Hybrid Couplers and Phase Shifters}

We consider two fundamental microwave components, hybrid couplers and phase shifters, together with two components needed to interconnect them: interconnections and permutation networks.
\begin{definition}
An interconnection is a $2$-port network with scattering matrix
\begin{equation}
\mathbf{S}=
\begin{bmatrix}
0 & 1\\
1 & 0
\end{bmatrix}.
\end{equation}
\end{definition}
\begin{definition}
A phase shifter with phase $\phi\in[0,2\pi)$ is a $2$-port network with scattering matrix
\begin{equation}
\mathbf{S}=
\begin{bmatrix}
0 & e^{j\phi}\\
e^{j\phi} & 0
\end{bmatrix}.
\end{equation}
\end{definition}
\begin{definition}
A permutation network reordering the signals at $N$ ports according to the permutation matrix $\mathbf{P}\in\{0,1\}^{N\times N}$ is a $2N$-port network with scattering matrix
\begin{equation}
\mathbf{S}=
\begin{bmatrix}
\mathbf{0} & \mathbf{P}^T\\
\mathbf{P} & \mathbf{0}
\end{bmatrix}.
\end{equation}
\end{definition}
\begin{definition}
A hybrid coupler is a $4$-port network with scattering matrix
\begin{equation}
\mathbf{S}=\frac{1}{\sqrt{2}}
\begin{bmatrix}
0 & 0 & j & 1\\
0 & 0 & 1 & j\\
j & 1 & 0 & 0\\
1 & j & 0 & 0
\end{bmatrix}.
\end{equation}
\end{definition}

All four components are lossless and reciprocal, since their scattering matrices are unitary and symmetric.
In addition, all of them are such that $\mathbf{S}_{11}=\mathbf{0}$ and $\mathbf{S}_{22}=\mathbf{0}$, i.e., there is no reflection at their ports.
We refer to a $2N$-port network with this property as a matched network.
When matched networks are also reciprocal, i.e., their scattering matrix is symmetric, they are fully characterized by their $\mathbf{S}_{21}$ block, which we denote as the transmission scattering matrix $\bar{\mathbf{S}}\in\mathbb{C}^{N\times N}$ \cite[Section~III]{ner26}.
A hybrid coupler has transmission scattering matrix $\bar{\mathbf{S}}=[[j,1]^T,[1,j]^T]/\sqrt{2}$, a phase shifter $\bar{s}=e^{j\phi}$, an interconnection $\bar{s}=1$, and a permutation network $\bar{\mathbf{S}}=\mathbf{P}$.
Since $\mathbf{v}=\mathbf{S}_{21}\mathbf{u}/2$, a \gls{milac} implemented with a matched network computes $\mathbf{v}=\bar{\mathbf{S}}\mathbf{u}/2$.

More complex networks are obtained by connecting networks in series and in parallel.
Two networks are connected in series when the output ports of the first are connected to the input ports of the second, and in parallel when their input and output ports become the input and output ports of a larger network, respectively.
The series and parallel of matched networks are matched networks whose transmission scattering matrices have a simple expression \cite[Corollaries~1,~2]{ner26}, as follows, which we adopt as definitions.
\begin{definition}
The series of two reciprocal matched networks with transmission scattering matrices $\bar{\mathbf{Q}}$ and $\bar{\mathbf{R}}$ is a reciprocal matched network with transmission scattering matrix $\bar{\mathbf{S}}=\bar{\mathbf{R}}\bar{\mathbf{Q}}$.
\label{def:series}
\end{definition}
\begin{definition}
The parallel of two reciprocal matched networks with transmission scattering matrices $\bar{\mathbf{Q}}$ and $\bar{\mathbf{R}}$ is a reciprocal matched network with transmission scattering matrix $\bar{\mathbf{S}}=\text{diag}(\bar{\mathbf{Q}},\bar{\mathbf{R}})$.
\label{def:parallel}
\end{definition}

The networks that can be built from the four components are then defined inductively, as follows.
\begin{definition}
A network is implementable with hybrid couplers and phase shifters if any of the following statements are true
\begin{enumerate}
\item The network is an interconnection,
\item The network is a phase shifter,
\item The network is a permutation network,
\item The network is a hybrid coupler,
\item The network is a series of two networks, both implementable with hybrid couplers and phase shifters,
\item The network is a parallel of two networks, both implementable with hybrid couplers and phase shifters.
\end{enumerate}
\label{def:impl}
\end{definition}

\subsection{Characterization of Implementable Networks}

The following theorem, which is the main result of \cite{ner26}, characterizes all networks implementable with hybrid couplers and phase shifters, and hence all the linear transformations that a \gls{milac} made of these components can compute.
\begin{theorem}
A network is implementable with hybrid couplers and phase shifters if and only if it is a reciprocal matched network with transmission scattering matrix $\bar{\mathbf{S}}\in\mathbb{C}^{N\times N}$ that can be decomposed as
\begin{equation}
\bar{\mathbf{S}}=\mathbf{P}_{L}\mathbf{D}_{L}\cdots\mathbf{P}_{2}\mathbf{D}_{2}\mathbf{P}_{1}\mathbf{D}_{1}\mathbf{P}_{0},\label{eq:theor}
\end{equation}
where $\mathbf{P}_{\ell}\in\mathbb{C}^{N\times N}$ is a permutation matrix, for $\ell=0,\ldots,L$, and $\mathbf{D}_{\ell}\in\mathbb{C}^{N\times N}$ is a block diagonal matrix, for $\ell=1,\ldots,L$, partitioned as
\begin{equation}
\mathbf{D}_{\ell}=\text{diag}\left(\mathbf{D}_{\ell,1},\ldots,\mathbf{D}_{\ell,C_\ell},e^{j\phi_{\ell,1}},\ldots,e^{j\phi_{\ell,S_\ell}}\right),\label{eq:Dell}
\end{equation}
where $\phi_{\ell,s}\in[0,2\pi)$, for $s=1,\ldots,S_\ell$, and $\mathbf{D}_{\ell,c}\in\mathbb{C}^{2\times2}$ is given by
\begin{equation}
\mathbf{D}_{\ell,c}=\frac{1}{\sqrt{2}}
\begin{bmatrix}
e^{j\theta_{\ell,c,11}} & e^{j\theta_{\ell,c,12}}\\
e^{j\theta_{\ell,c,21}} & e^{j\theta_{\ell,c,22}}
\end{bmatrix},\label{eq:Dellc}
\end{equation}
with $\theta_{\ell,c,11},\theta_{\ell,c,12},\theta_{\ell,c,21}\in[0,2\pi)$ and $\theta_{\ell,c,22}=\pi-\theta_{\ell,c,11}+\theta_{\ell,c,12}+\theta_{\ell,c,21}$, for $c=1,\ldots,C_\ell$, where $2C_\ell+S_\ell=N$.
\label{the:1}
\end{theorem}
\begin{proof}
Please, refer to \cite[Section~IV]{ner26}.
\end{proof}

Theorem~\ref{the:1} was used in \cite{ner26} to show that the $N\times N$ \gls{dft}, Hadamard, and Haar matrices are implementable with hybrid couplers and phase shifters for any $N$ power of two.
In the following, we formalize all the above definitions and Theorem~\ref{the:1} in Lean, together with the implementability of the \gls{dft}.
The complete code is available in the repository linked in the abstract.

\section{Lean Formalization: Definitions}

Lean is a functional programming language and a proof assistant, in which mathematical statements can be written and their proofs verified by a trusted kernel \cite{dem21}.
Our formalization builds on Mathlib \cite{the20}, the Lean library that provides widely used mathematical concepts such as matrices, permutations, and the complex exponential.
To trust that the proofs are correct, a reader does not need to check them, but only that the definitions and the statements of the theorems presented in the following capture the intended concepts.
The keyword \texttt{noncomputable} only indicates that a definition involves real numbers or the complex exponential, which Lean cannot evaluate exactly, and has no effect on the proofs.

We define a network as an $n\times n$ complex matrix.
Physically, this matrix represents the transmission scattering matrix of the microwave network having $n$ input and $n$ output ports.
Since all the considered components are reciprocal matched networks, and reciprocal matched networks are closed under series and parallel connections, the transmission scattering matrix is sufficient to describe all the networks of interest.
The following abbreviation is therefore used throughout the whole Lean script, simply meaning that a network is an $n\times n$ complex matrix.
\begin{lstlisting}
abbrev Network (n : ℕ) :=
  Matrix (Fin n) (Fin n) ℂ
\end{lstlisting}
An interconnection is defined as the $1\times 1$ matrix $1$.
\begin{lstlisting}
def interconnection : Network 1 :=
  1
\end{lstlisting}
A phase shifter is defined as the $1\times 1$ matrix $e^{j\phi}$.
\begin{lstlisting}
noncomputable def phaseShifter (φ : ℝ) :
    Network 1 :=
  fun _ _ => Complex.exp (φ * Complex.I)
\end{lstlisting}
Note that the phase shift $\phi$ is any real number, rather than being restricted to $[0,2\pi)$.
This is without loss of generality, since $e^{j\phi}$ is $2\pi$-periodic.
A permutation network is defined as an $n\times n$ permutation matrix, constructed by Mathlib from a permutation $\sigma$ of the ports.
\begin{lstlisting}
def permutationNetwork {n : ℕ}
    (σ : Equiv.Perm (Fin n)) : Network n :=
  Equiv.Perm.permMatrix (R := ℂ) σ
\end{lstlisting}
A hybrid coupler is defined as the $2\times 2$ matrix $[[j,1]^T,[1,j]^T]/\sqrt{2}$.
\begin{lstlisting}
noncomputable def hybridCoupler : Network 2 :=
  (1 / Real.sqrt 2 : ℂ) • !![Complex.I, 1;
                             1, Complex.I]
\end{lstlisting}

The series of two networks \texttt{A} and \texttt{B} with $n$ input and output ports, where \texttt{A} is followed by \texttt{B}, is defined as the matrix product \texttt{B * A}, following Definition~\ref{def:series}.
\begin{lstlisting}
def series {n : ℕ} (A B : Network n) :
    Network n :=
  B * A
\end{lstlisting}
The parallel of a network \texttt{A} with $n$ input and output ports and a network \texttt{B} with $m$ input and output ports is defined as their block diagonal matrix, following Definition~\ref{def:parallel}.
\begin{lstlisting}
def parallel {n m : ℕ} (A : Network n)
    (B : Network m) : Network (n + m) :=
  Matrix.reindex finSumFinEquiv finSumFinEquiv
      (Matrix.fromBlocks A 0 0 B)
\end{lstlisting}

The notion of implementable network in Definition~\ref{def:impl} is formalized as an inductive predicate \texttt{Implementable}, with one constructor for each of its six statements.
A network is implementable if and only if it can be obtained by applying these six rules a finite number of times.
\begin{lstlisting}
inductive Implementable : {n : ℕ} → Network n
    → Prop where
  | ic : Implementable interconnection
  | ps (φ : ℝ) :
      Implementable (phaseShifter φ)
  | pn {n : ℕ} (σ : Equiv.Perm (Fin n)) :
      Implementable (permutationNetwork σ)
  | hc : Implementable hybridCoupler
  | series {n : ℕ} {A B : Network n} :
      Implementable A → Implementable B →
      Implementable (series A B)
  | parallel {n m : ℕ}
      {A : Network n} {B : Network m} :
      Implementable A → Implementable B →
      Implementable (parallel A B)
\end{lstlisting}
Since this definition is inductive, a property can be proved for all implementable networks by induction, i.e., by proving it for the four components and showing that it is preserved by series and parallel connections.
This is the key tool in the proof of Theorem~\ref{the:1}.

\section{Lean Formalization: Main Result}

We begin by defining as \texttt{layered} a network that satisfies the condition in Theorem~\ref{the:1}.
This is done by first introducing more basic building blocks.
The first building block is the diagonal matrix of phases
\begin{equation}
\text{diag}\left(e^{j\phi_{1}},\ldots,e^{j\phi_{n}}\right),\label{eq:phaseDiag}
\end{equation}
representing $n$ phase shifters in parallel, which is defined as follows.
\begin{lstlisting}
noncomputable def phaseDiagonal {n : ℕ}
    (φ : Fin n → ℝ) : Network n :=
  Matrix.diagonal (fun i =>
      Complex.exp (φ i * Complex.I))
\end{lstlisting}
The second building block is the $2\times 2$ block in~\eqref{eq:Dellc}, i.e.,
\begin{equation}
\frac{1}{\sqrt{2}}
\begin{bmatrix}
e^{j\theta_{11}} & e^{j\theta_{12}}\\
e^{j\theta_{21}} & e^{j(\pi-\theta_{11}+\theta_{12}+\theta_{21})}
\end{bmatrix},\label{eq:block}
\end{equation}
which depends on three parameters $\theta_{11}$, $\theta_{12}$, and $\theta_{21}$, and is defined as follows.
\begin{lstlisting}
noncomputable def block ( θ₁₁  θ₁₂  θ₂₁ : ℝ) :
    Network 2 :=
  (1 / Real.sqrt 2 : ℂ) • !![
    Complex.exp ( θ₁₁ * Complex.I),
    Complex.exp ( θ₁₂ * Complex.I);
    Complex.exp ( θ₂₁ * Complex.I),
    Complex.exp ((Real.pi - θ₁₁ + θ₁₂ + θ₂₁ )
        * Complex.I)]
\end{lstlisting}
These two building blocks are combined into the block diagonal matrix $\mathbf{D}_\ell$ in \eqref{eq:Dell}, made of $C$ blocks \eqref{eq:block} and $S$ phases, which is a network with $2C+S$ input and output ports.
The $C$ blocks are assembled with Mathlib's \texttt{Matrix.blockDiagonal}, and placed in parallel with the diagonal matrix of phases.
\begin{lstlisting}
noncomputable def blockDiagonal {C S : ℕ}
    ( θ₁₁  θ₁₂  θ₂₁ : Fin C → ℝ) (φ : Fin S → ℝ) :
    Network (2 * C + S) :=
  parallel (Matrix.reindex finProdFinEquiv
      finProdFinEquiv (Matrix.blockDiagonal
      fun c => block ( θ₁₁ c) ( θ₁₂ c) ( θ₂₁ c)))
      (phaseDiagonal φ)
\end{lstlisting}
Note that, with this indexing, the two ports of the $c$th block are $c$ and $C+c$, rather than two adjacent ports as in \eqref{eq:Dell}.
The two orderings differ by a permutation, which is absorbed by the adjacent permutation matrices, so the set of matrices of the form \eqref{eq:theor} is unchanged.

A layer $\mathbf{P}_\ell\mathbf{D}_\ell$ of the decomposition \eqref{eq:theor} is described by the structure \texttt{Layer}, which collects the number of blocks $C_\ell$, the number of phases $S_\ell$, the parameters of the $C_\ell$ blocks and the $S_\ell$ phases, the constraint $2C_\ell+S_\ell=n$, and the permutation $\sigma$ defining $\mathbf{P}_\ell$.
\begin{lstlisting}
structure Layer (n : ℕ) where
  C : ℕ
  S : ℕ
  θ₁₁ : Fin C → ℝ
  θ₁₂ : Fin C → ℝ
  θ₂₁ : Fin C → ℝ
  φ : Fin S → ℝ
  size : 2 * C + S = n
  σ : Equiv.Perm (Fin n)
\end{lstlisting}
Finally, the product $\mathbf{P}_{L}\mathbf{D}_{L}\cdots\mathbf{P}_{1}\mathbf{D}_{1}\mathbf{P}_{0}$ in \eqref{eq:theor} is defined recursively from the initial permutation $\sigma_0$, defining $\mathbf{P}_0$, and a list of layers, ordered from $L$ down to $1$.
An empty list gives $\mathbf{P}_0$, while a list with head $\mathbf{P}_\ell\mathbf{D}_\ell$ gives $\mathbf{P}_\ell\mathbf{D}_\ell$ multiplied by the product of the remaining layers.
The number of layers $L$ is the length of the list.
\begin{lstlisting}
noncomputable def layered {n : ℕ}
    ( σ₀ : Equiv.Perm (Fin n)) :
    List (Layer n) → Network n
  | [] => permutationNetwork σ₀
  | layer :: layers =>
      permutationNetwork layer.σ *
      Matrix.reindex (finCongr layer.size)
      (finCongr layer.size)
      (blockDiagonal layer. θ₁₁  layer. θ₁₂
      layer. θ₂₁  layer.φ) * layered σ₀ layers
\end{lstlisting}

The sufficient condition of Theorem~\ref{the:1} states that any \texttt{layered} network is implementable.
\begin{lstlisting}
theorem implementable_if_layered {n : ℕ}
    ( σ₀ : Equiv.Perm (Fin n))
    (layers : List (Layer n)) :
    Implementable (layered σ₀ layers) := by
  ...
\end{lstlisting}
The statement of this sufficient condition is followed by the keyword \texttt{:= \textcolor{leanblue}{by}}, after which Lean expects the proof.
This proof is omitted here for brevity, but it can be found in the GitHub repository, where it has been checked by Lean.
As in~\cite[Section~IV-A]{ner26}, the key step is that each block \eqref{eq:block} is a series of three implementable networks, i.e., the parallel of a phase shifter and an interconnection, a hybrid coupler, and the parallel of two phase shifters.
Then, the diagonal matrix of phases \eqref{eq:phaseDiag} is shown to be implementable by induction on $n$, and consequently also the block diagonal matrix $\mathbf{D}_\ell$ is implementable.
Finally, a \texttt{layered} network is shown to be implementable, using the fact that the series of implementable networks is implementable, and that permutation networks are implementable.

The necessary condition of Theorem~\ref{the:1} states that any implementable network is \texttt{layered}.
\begin{lstlisting}
theorem implementable_only_if_layered {n : ℕ}
    {S : Network n} (hS : Implementable S) :
    ∃ σ₀ : Equiv.Perm (Fin n),
    ∃ layers : List (Layer n),
    S = layered σ₀ layers := by
  ...
\end{lstlisting}
Also for this theorem we omit the proof after the keyword \texttt{:= \textcolor{leanblue}{by}}, but it can be found in the GitHub repository, where it has been checked by Lean.
In \cite[Section~IV-B]{ner26}, this is proved by dividing a generic network into $L$ subnetworks, each containing hybrid couplers and phase shifters in parallel between two permutation networks.
This argument is intuitive, but it is not directly expressed in terms of Definition~\ref{def:impl}.
The Lean proof is instead by induction on Definition~\ref{def:impl}, and considers the six cases separately.
\begin{enumerate}
\item A phase shifter is a \texttt{layered} network with a single layer containing one phase and no blocks.
\item An interconnection is a \texttt{layered} network since it is a phase shifter with zero phase.
\item A permutation network is a \texttt{layered} network with only $\mathbf{P}_0$ and no layers.
\item A hybrid coupler is a \texttt{layered} network with a single layer containing the block \eqref{eq:block} with $\theta_{11}=\pi/2$ and $\theta_{12}=\theta_{21}=0$.
\item The series of two \texttt{layered} networks is \texttt{layered}, since their layers can be concatenated, merging the last permutation matrix of the first network with the first permutation matrix of the second into a single permutation matrix.
\item The parallel of two \texttt{layered} networks \texttt{A} and \texttt{B} is \texttt{layered}, since it can be written as the series of the parallel of \texttt{A} with an identity and the parallel of an identity with \texttt{B}.
The parallel of a \texttt{layered} network with an identity is \texttt{layered}, since each of its layers can be extended with zero phases, i.e., interconnections.
\end{enumerate}

By combining the sufficient and necessary conditions, we obtain Theorem~\ref{the:1} in Lean, which states that a network is implementable if and only if it is \texttt{layered}.
\begin{lstlisting}
theorem implementable_iff_layered {n : ℕ}
    {S : Network n} :
    Implementable S ↔
    ∃ σ₀ : Equiv.Perm (Fin n),
    ∃ layers : List (Layer n),
    S = layered σ₀ layers := by
  constructor
  · exact implementable_only_if_layered
  · rintro ⟨ σ₀ , layers, rfl⟩
    exact implementable_if_layered σ₀ layers
\end{lstlisting}
The theorem holds for any number of ports $n$ and any number of layers $L$, and its proof is certified by the Lean kernel.

\balance

\section{Lean Formalization: Implementing the DFT}

We now formalize the proof that the \gls{dft} can be computed by a \gls{milac} made of hybrid couplers and phase shifters, i.e., that the $N\times N$ \gls{dft} matrix is implementable for any $N$ power of two \cite[Corollary~3]{ner26}.
The $N\times N$ \gls{dft} matrix is defined as
\begin{equation}
[\mathbf{F}_N]_{i,k} = \frac{1}{\sqrt{N}}\omega^{ik},
\end{equation}
for $i,k=0,\ldots,N-1$, where $\omega = e^{-j 2\pi/ N}$.
Note that the indices start from $0$ here, in agreement with Lean.
\begin{lstlisting}
noncomputable def dft (N : ℕ) : Network N :=
  fun i k => (1 / Real.sqrt N : ℂ) *
      Complex.exp (-2 * Real.pi * Complex.I *
      (i : ℕ) * (k : ℕ) / N)
\end{lstlisting}
The following theorem states that the $2^L\times 2^L$ \gls{dft} matrix is implementable, for any $L\in\mathbb{N}$.
\begin{lstlisting}
theorem implementable_dft (L : ℕ) :
    Implementable (dft (2 ^ L)) := by
  ...
\end{lstlisting}
In \cite[Section~V]{ner26}, this result is obtained by deriving a decomposition of $\mathbf{F}_N$ of the form \eqref{eq:theor}.
The Lean proof instead proceeds by induction on $L$, following the radix-2 structure of the \gls{fft}.
For $L=0$, $\mathbf{F}_1=1$ is an interconnection.
For the inductive step, the $2^{L+1}\times 2^{L+1}$ \gls{dft} matrix is shown to be, up to a permutation of its input ports and a permutation of its output ports, the series of two $2^{L}\times 2^{L}$ \gls{dft} matrices in parallel, followed by $2^{L}$ butterfly matrices $[[1,1]^T,[\omega^{i},-\omega^{i}]^T]/\sqrt{2}$ in parallel, for $i=0,\ldots,2^{L}-1$, where $\omega=e^{-j2\pi/2^{L+1}}$.
Each butterfly matrix is a block \eqref{eq:block} with suitable phases, and is therefore implementable.
The two smaller \gls{dft} matrices are implementable by the induction hypothesis, and the permutations of the input and output ports are implemented by two permutation networks, which concludes the proof.
By Theorem~\ref{the:1}, \texttt{implementable\_dft} also guarantees that $\mathbf{F}_{2^L}$ admits a decomposition of the form~\eqref{eq:theor}, although it does not identify the specific matrices $\mathbf{P}_\ell$ and $\mathbf{D}_\ell$ given in \cite[Proposition~3]{ner26}.

\section*{Acknowledgment}

The authors used Claude (Anthropic) to assist in developing the Lean code, whose proofs are all checked by the Lean kernel.
The authors also thank the Mathlib community for developing and maintaining the Mathlib library \cite{the20}, on which this formalization builds.

\section{Conclusion}

We have presented a first formalization of analog computing with microwaves in Lean.
We have formalized hybrid couplers, phase shifters, interconnections, and permutation networks, together with their series and parallel connections, and the class of linear transforms they can compute.
Specifically, we have formally proved the necessary and sufficient condition characterizing these transforms, derived in \cite{ner26}, and the implementability of the \gls{dft} of any size power of two.

This work opens several directions for future research.
\textit{First}, the implementability of the Hadamard and Haar transforms \cite{ner26}, and maybe other transforms, can be formalized.
\textit{Second}, the scattering matrices of generic networks (not necessarily matched) and their series and parallel connections can be formalized.
\textit{Third}, other microwave components can be included to characterize the functions they can compute.
\textit{Fourth}, the formalized results can be used to automatically generate verified designs of microwave networks computing a desired transformation.

\bibliographystyle{IEEEtran}
\bibliography{IEEEabrv,main}

\end{document}